\documentclass[12pt]{article}
\usepackage{amsmath}
\usepackage{graphicx}
\usepackage{enumerate}
\usepackage{natbib}
\usepackage{url} 
\usepackage{color}
\usepackage{array}
\usepackage{float}
\usepackage{multirow}
\usepackage{amsmath}
\usepackage{amsthm}
\usepackage{amssymb,mathrsfs}
\RequirePackage[colorlinks,citecolor=blue,urlcolor=blue]{hyperref}
\usepackage{breakurl}
\def\wh{\widehat}
\def\wt{\widetilde}
\def\ol{\overline}

\makeatletter

\numberwithin{equation}{section}
\theoremstyle{definition}

\theoremstyle{plain}
\newtheorem{thm}{\protect\theoremname}[section]
\theoremstyle{plain}

\numberwithin{equation}{section}
\usepackage[left=3.5cm, right=3.5cm, bottom = 2.5cm]{geometry}
\date{}

\allowdisplaybreaks[4]

\usepackage{graphics}

\makeatother

\providecommand{\conditionname}{Condition}
\providecommand{\lemmaname}{Lemma}
\providecommand{\theoremname}{Theorem}
\newcommand{\blind}{0}

\begin{document}

\def\spacingset#1{\renewcommand{\baselinestretch}%
{#1}\small\normalsize} \spacingset{1}


\if0\blind
{
  \title{\bf {Scalable Estimation and Inference with 
  		Large-scale or Online Survival Data
  	}
  }
  \author{    Jinfeng Xu \\
    Department of Statistics and Actuarial Science\\
    The University of Hong Kong\\
    and \\
    Zhiliang Ying
    \\
    Department of Statistics\\
    Columbia University\\
    and \\
    Na Zhao \\
    Department of Statistics and Actuarial Science\\
    The University of Hong Kong   
	}
  \maketitle
} \fi

\if1\blind
{
  \bigskip
  \bigskip
  \bigskip
  \begin{center}
    {\LARGE\bf Scalable inference with large-scale or online survival data
\end{center}
  \medskip
} \fi

\begin{abstract}
With the rapid development of data collection and aggregation technologies in many scientific disciplines, it is becoming increasingly ubiquitous to conduct large-scale or online regression to analyze real-world data and unveil real-world evidence.  In such applications, it is often numerically challenging or sometimes infeasible to store the entire dataset in memory. Consequently, classical batch-based estimation methods that involve the entire dataset are less attractive or no longer applicable. Instead, recursive estimation methods such as stochastic gradient descent that process data points sequentially  are more appealing, exhibiting both numerical convenience and memory efficiency. In this paper, for scalable estimation with randomly 
censored large or online survival data in the accelerated failure time model, we propose a stochastic gradient descent method which recursively updates the estimates in an online manner as data points arrive sequentially in streams.
Theoretical results such as asymptotic normality and estimation efficiency are established to justify its validity. Furthermore, to quantify the uncertainty associated with the proposed stochastic gradient descent estimator and facilitate statistical inference, we develop a scalable resampling strategy that specifically caters to the large-scale or online setting.  Simulation studies and a real data application are also
	provided to assess its performance and illustrate its practical utility.
\end{abstract}

\noindent%
\vfill

\newpage
\spacingset{1.5} 

\section{Introduction}
The volume and velocity of information about individual patients or customers are greatly increasing with use of electronic records and personal device. Potential benefits of utilizing such information could be numerous, ranging from the ability to determine large-scale effects of treatment to the ability to monitor real-time effects of treatment on a general population. In the context of this wealth of real-world data (RWD), it is often necessary to conduct large-scale or online regression to unveil real-world evidence (RWE). For large-scale or online survival data, the response variable is survival time and is subject to possible right censoring. Here online
data refers to the situation where the observations arrive sequentially in time and once they are processed, they are discarded for memory or privacy reasons. 
Let $T$ be the survival time, $C$ the censoring time, and $X$ the $p$-vector of covariates. Define $\widetilde{T}=T\wedge C$, $\delta=I(T\le C)$, and $Z=(\widetilde{T},\delta, X)$. Suppose that the data consist of independent and identically distributed (i.i.d.) copies of $Z$ and denote it 
by $\mathcal{D}_N=\{Z_1,...,Z_N\}$ for large-scale data and $\mathcal{D}=\{Z_1,Z_2,...\}$ for online data, respectively. Here $N$ denotes the size of the dataset and is assumed to be large. We consider 
the accelerated failure time model which postulates \begin{equation}
\label{aft}
\log T=X^T\beta_0+\epsilon,
\end{equation}
where $\beta_0$ is a $p$-vector of unknown regression parameters, the stochastic error $\epsilon$ is independent of $X$, and its distribution is left unspecified. Because it provides a natural formulation of the effects of covariates on potentially censored response variable,
the model (\ref{aft}), along with the Cox proportional hazards model, are  two main approaches to the regression analysis of 
censored data \citep{kalbfleisch2002statistical,zeng2007efficient}. Define the residual $e_{i}(\beta)=\log \widetilde{T}_{i}-\beta^TX_{i}$. Let $N_{i}(\beta ; t)=\delta_{i} I\left(e_{i}(\beta) \leqslant t\right) \text { and } Y_{i}(\beta ; t)=I\left(e_{i}(\beta) \geqslant t\right)$
denote the counting and at risk processes of the residual, respectively.
Write $$
S^{(0)}(\beta ; t)=N^{-1} \sum_{i=1}^{N} Y_{i}(\beta ; t), \quad S^{(1)}(\beta ; t)=N^{-1} \sum_{i=1}^{N} Y_{i}(\beta ; t) X_{i}.
$$ 
For classical batch-based methods, the estimation and inference in model 
(\ref{aft})  often centers on solving
the weighted rank-based estimating equations which take the form 
\begin{equation}
\label{estimation}
U_{\phi}(\beta)=\sum_{i=1}^{N} \int_{-\infty}^{\infty} \phi(\beta ; t)\left\{X_{i}-\overline{X}(\beta ; t)\right\} d N_{i}(\beta ; t),
\end{equation}
where $
\overline{X}(\beta ; t)=S^{(1)}(\beta ; t) / S^{(0)}(\beta ; t), \text { and } \phi
$ is a possibly data-dependent weight function. The choices of $\phi=1$ and $\phi=S^{(0)}$ correspond
to the log-rank \citep{mantel1966evaluation} and Gehan statistics \citep{gehan1965generalized}, respectively. In particular, with $\phi=S^{(0)}$,
\begin{equation}
\label{eq:gehan}
S_N(\beta)=N^{-1} \sum_{i=1}^{N} \sum_{j=1}^{N} \delta_{i}\left(X_{i}-X_{j}\right) I{\left\{e_{i}(\beta) \leqslant e_{j}(\beta)\right\}}.
\end{equation}
Solving (\ref{eq:gehan}) is 
equivalent to minimizing the objective function 
\begin{equation}
L_N(\beta)=N^{-1} \sum_{i=1}^{N} \sum_{j=1}^{N} \delta_{i}\left\{e_{i}(\beta)-e_{j}(\beta)\right\}^{-}, \label{obj:gehan}
\end{equation}
where  \(a^{-}=|a| I{\{a<0\}}\).
The optimization can be formulated as a linear programming
problem and solved by standard statistical packages \citep{jin2003rank,koenker2005quantile,chiou2014fitting}. 
This approach yields numerically 
efficient estimation and inference procedures 
for the accelerated failure time model.
However, it requires the entire dataset to be stored in memory and the computational complexity is of order $O(N^2)$. 
When $N$ is extremely large, as in large-scale RWD, or in an online setting, as in streaming RWD, these batch-based estimation methods become numerically infeasible. 
Instead, online learning tools avoid the problem of managing the large-scale data exceeding the size of the memory and are applicable to the streaming data where the observations arrive sequentially by sharing the property of analyzing one observation at a time \citep{bottou2005online}.
As an online learning tool, stochastic gradient descent (SGD) algorithms have recently regained a great deal of attention in the statistical community for analyzing big data since nowadays it is becoming increasingly prevalent in practice to manage and process big data that are much larger than the memory of a typical PC \citep{bottou2010large}. Stochastic gradient descent, as a stochastic approximation method,  processes one data point at a time upon its arrival. For example, suppose that we have $N$ i.i.d. observations, $\beta$ denotes the model 
parameter, and  $g_i(\beta)$ is minus log likelihood of the $i$th observation, $i=1,...,N$. 
The maximum likelihood estimates of $\beta$ can then be
solved by minimizing the objective function
\begin{equation}
g(\beta)=\frac{1}{N}\sum_{i=1}^{N}g_i(\beta).\label{eq: general-like}
\end{equation}
Instead of using the Newton-Raphson algorithm to directly minimize $g(\beta)$, the SGD method calculates the estimates by 
recursively updating the estimates  upon the arrival of each observation, starting with some initial estimates $\wh{\beta}_{0}$, 
for $n=1,...,N$,
\begin{equation}
\wh{\beta}_n=\wh{\beta}_{n-1}-\gamma_n \nabla g_n(\wh{\beta}_{n-1}),\label{eq: general-sgd}
\end{equation}
where $\gamma_n$ is some learning rate. This approach provides a numerically convenient and memory efficient approach for large-scale or online applications. The estimator $\wh{\beta}_n$ or its variants has been shown to exhibit good properties 
such as asymptotic consistency and normality under some regularity conditions \citep{ruppert1988efficient, polyak1992acceleration}. The SGD algorithms have been successfully applied to implement linear regression, logistics regression, and robust regression; see for example \cite{moulines2011non} and \cite{fang2018online}. 
It is important to note that stochastic gradient descent is not directly applicable to the
estimation based on (\ref{eq:gehan}) because it involves the calculation of the ranks of $e_i(\beta), i=1,...,N$ and upon the arrival of each data point, 
the recursive updating cannot be carried out without resorting to the entire dataset. To be more
specific, denote the gradient 
of the objective function (\ref{obj:gehan}) that involves the observation $i$ by
$\nabla L_i(\beta)$. We find that $L_i(\beta)=\delta_{i}\sum\limits_{j=1}^{N} \left\{e_{i}(\beta)-e_{j}(\beta)\right\}^{-}$ and it involves all the observations in the set $\mathcal{R}_i(\beta)=\{j: e_{i}(\beta) \leq e_{j}(\beta)\}$. To evaluate $\nabla L_i(\beta)$, we need to have the entire set of observations. However, having the entire set of observations is what the SGD method attempts to avoid. 
To address this problem and overcome the difficulty of the original SGD method, we propose a new strategy which 
retains the numerical simplicity of SGD in recursive 
and online updating as well as caters to the specific nature of rank-based analysis of survival data.
Therefore, the proposed estimation procedure scales well for large-scale and streaming survival data.
 Furthermore, apart from scalable estimation, there remains the core inferential need to assess the quality and quantify the uncertainty of the proposed SGD stimator. The ability to assess estimator quality efficiently is essential to allow efficient use of available resources by processing only as much data as is necessary to achieve a desired accuracy or confidence. 
We propose an online resampling method which allows scalable inference in an online and parallel manner.
It preserves the automatic nature of the original bootstrap and
is thus applicable to a wide variety of inferential problems. 
The rest of this paper is organized as follows.  In Section 2, we 
propose a  scalable SGD method for large-scale or online survival data and study the asymptotic properties
of the proposed estimator. In Section 3,  we propose an online resampling strategy and establish the theory to justify its validity.  Simulation studies and an application to the Surveillance, Epidemiology, and End Results (SEER) breast cancer data are provided in Section 4 to examine the performance of the proposed method.
All the proofs are presented in the Appendix.

\section{Method}
\subsection{Estimation}
Because the summands in (\ref{estimation}) or (\ref{eq:gehan}) involve more than one data point, the original SGD approach which sequentially makes use of one data point for each recursive 
updating is no longer applicable. To address this problem, we propose a  new stochastic gradient descent method which sequentially updates the estimates upon the arrival of every $k$ data points. Here $k$ is a fixed integer greater than $1$ to allow the recursive updating to be effectively carried out based on (\ref{estimation}) or (\ref{eq:gehan}). For ease of exposition, we consider the large-scale setting 
but the method can be applied to the online setting in a similar fashion.  Without loss of generality, it is assumed that $N=nk$. The
data $\mathcal{D}_N=\{Z_1,...,Z_N\}$ 
consist of $N$ i.i.d. copies of $Z=(\wt{T},\delta,X)$. Let $D_i=\{(\wt{T}_{i,l}, \delta_{i,l}, X_{i,l}),1\leq l\leq k\}$ denote
the $k$ data points in the $i$th updating, $i=1,...,n$. Based on $D_i$, it can be shown that the true regression parameter $\beta_0$ minimizes the expectation of the 
objective function
\begin{eqnarray}
l_{i}(\beta)=\frac{1}{k}\sum_{l=1}^{k}\sum_{j=1}^{k} \delta_{i,l}\left\{e_{i,l}(\beta)-e_{i,j}(\beta)\right\}^{-}, \label{min: obj}
\end{eqnarray}
which has the gradient
\begin{eqnarray}
s_{i}(\beta)=\frac{1}{k}\sum_{l=1}^{k}\sum_{j=1}^{k}\delta_{i,l}\left(X_{i,l}-X_{i,j}\right) I{\left\{e_{i,l}(\beta) \leqslant e_{i,j}(\beta)\right\}}.\label{eq: likeli-si}
\end{eqnarray}
In the same spirit of the original SGD, we propose the following 
recursive algorithm. Starting from some initial estimate $\wh{\beta}_0$, for $i=1,...,n$,
 we update the estimate  via
\begin{eqnarray}
\wh{\beta}_i=\wh{\beta}_{i-1}-\gamma_i s_{i}(\wh{\beta}_{i-1}),\label{eq: batch-sgd}
\end{eqnarray}
where the learning rates are $\gamma_i=\gamma_1 i^{-\alpha}$ with $\gamma_1>0$ and $\alpha\in (0.5, 1)$.
Furthermore,  as suggested by \cite{ruppert1988efficient} and \cite{polyak1992acceleration}, we consider the averaging estimate,
\begin{eqnarray}\label{SGD-avg}
\ol{\beta}_n=\frac{1}{n}\sum_{i=1}^n\wh{\beta}_i,
\end{eqnarray}
which can also be recursively updated given that $\ol{\beta}_i=(i-1)\ol{\beta}_{i-1}/i+\wh{\beta}_i/i$, $i=1,...,n$.
It can be seen that the proposed  method retains the appealing properties of SGD  such as numerical convenience and memory efficiency. It scales well with the size of the dataset and is readily applicable to large-scale and online applications. Next we study its asymptotic properties. 
\subsection{Limiting distribution}
It is assumed that $C$ is independent of  $T$ conditional on $X$. Throughout the paper, we shall use  $F$, $f$ and  $\overline{F}=1-F$
to denote the distribution, density and survival functions of $\varepsilon$, respectively. The conditional distribution, 
density and survival functions of $\log C$ given $X$ are denoted by $G(\cdot|X), g(\cdot|X)$ and $\overline{G}(\cdot|X)=1-G(\cdot|X)$,
respectively. Let $L(\beta)=\mathbb{E} l_{i}(\beta)$, $S(\beta)=\mathbb{E} s_{i}(\beta)$ and 
$H(\beta)=\nabla^2 L(\beta)$ be the Hessian matrix of $L(\beta)$. Define  $H_0=H(\beta_0)$
and $V=\mbox{Cov}(s_i(\beta_0))$. 

First we introduce the following assumptions.
\begin{enumerate}	
	\item[(A1)] The covariate $X$ is bounded and the matrix $\mbox{Cov}(X)$ is full rank.
	
	\item[(A2)]  The error density function 
	$f$ and its derivative $f^{'}$  
	are bounded.
	 
	\item[(A3)] The conditional density function 
	of  $C$ and its derivative 
	are bounded.
	
	\item[(A4)] The matrix $H_0$ is strictly positive definite.
	
	
\end{enumerate}
We establish the asymptotic normality of the SGD estimator $\ol{\beta}_n$ as follows.
\begin{thm}\label{th:aymnorma}
	Let $k$ be a fixed integer greater than 1. Under Assumptions A1-A4, as $n\rightarrow \infty$, we have (i)
	\begin{eqnarray}\label{linearrep}
		\sqrt{n} (\ol{\beta}_n-\beta_0) =-n^{-1/2} \sum _{i=1}^n H_0^{-1} s_{i}(\beta_0) + o_p(1).
	\end{eqnarray}
	and (ii)
	\begin{eqnarray}\label{asymdist}
		\sqrt{n} (\ol{\beta}_n-\beta_0) \Rightarrow N(0, \Sigma),
	\end{eqnarray}
	where $\Sigma=H_0^{-1} V H_0^{-1}$, $V={\normalfont \mbox{Cov}}(s_i(\beta_0))$,
	and  $s_i(\beta)$ is defined in (\ref{eq: likeli-si}).
\end{thm}
\noindent {\bf Remark 1.} It is important to note that in practice,
the proposed SGD is run only $n$ steps. Therefore, $n$ is 
required to be large to ensure that the SGD  runs until convergence and the above theoretical results
hold.

\subsection{Asymptotic relative efficiency}\label{are}
Let $\wt{\beta}_N$ denote the classical batch-based estimator which minimizes the objective function (\ref{obj:gehan}) involving the entire dataset. Let
$\Gamma_0(u)=\mathbb{E}[\overline{G}(u+X^T\beta_0|X)]$,
$\Gamma_1(u)=\mathbb{E}[X\overline{G}(u+X^T\beta_0|X)]$,
and $\Gamma_2(u)=\mathbb{E}[X^{\otimes 2}\overline{G}(u+X^T\beta_0|X)]$.
 As $N\rightarrow \infty$,  it is known that 
 $\sqrt{N}(\wt{\beta}_N-\beta_0)$ is 
asymptotically normally distributed with mean zero and the covariance 
$A^{-1}BA^{-1}$, where
\begin{equation}
A=\int_{-\infty}^{\infty}\ol{F}(u)[\Gamma_2(u)\Gamma_0(u)-\Gamma_1(u)\Gamma_1^T(u)]\frac{\lambda^{'}(u)}{\lambda(u)}dF(u),\quad  \label{eq: V}
\end{equation}
\begin{equation}
B=\int_{-\infty}^{\infty}[\ol{F}(u)\Gamma_0(u)]^2[\Gamma_2(u)-\frac{\Gamma_1(u)\Gamma_1^T(u)}{\Gamma_0(u)}]dF(u),\quad
\end{equation}
and
$\lambda(\cdot)$ is the hazard function of $\epsilon$, and 
$\lambda^{'}(u)=d \lambda(u)/du$ \citep{tsiatis1990estimating, ying1993large}.
Therefore both the classical batched-based  method  and the proposed SGD method yield  asymptotically unbiased and normally distributed estimates.
For any given $a\in \mathbb{R}^p$, when comparing these two methods in estimating $a^T\beta$,  a measure of asymptotic relative efficiency of 
the proposed SGD method relative to the classical batch-based method can be defined as
$$RE(k)=\frac{1}{k}\frac{a^TA^{-1}BA^{-1}a}{a^TH_0^{-1}VH_0^{-1}a}.$$ 
\begin{thm}\label{th:aeff}
Under Assumptions A1-A4, for any $a\in \mathbb{R}^p$,  we have $RE(k)\le 1$.
Furthermore, when $k\rightarrow \infty$,  $RE(k)=1-O(\frac{1}{k})\rightarrow 1$.
\end{thm}
\noindent The above theorem suggests the efficiency of the proposed method when $k$ is large, as well expected.
We also use Monte-carlo method to evaluate $RE(k)$ numerically in Section \ref{simstudy}.

\section{Inference}
From Theorem \ref{th:aymnorma}, we can conduct statistical inference based on $\ol{\beta}_n$ provided that we can estimate the covariance matrix $H_0^{-1}VH_0^{-1}$. Because $H_0$ 
involves the unknown hazard function, to bypass the difficulty of nonparametric smoothing in estimating the hazard function, we can use some resampling procedure to approximate the sampling distribution of $\sqrt{n}(\ol{\beta}_n-\beta_0)$. 
We propose an online bootstrap resampling procedure, which recursively updates the  SGD estimate as well as a large number of randomly perturbed SGD estimates, upon the arrival of every $k$ observations. The resampling 
strategy based on the random perturbation has also been widely used for inference 
in classical batch-based methods \citep{rao1992approximation,jin2003rank,peng2008survival}.
 Specifically, let $\Omega=\left\{\omega_{1}, \cdots, \omega_{n
}\right\}$ be a set of i.i.d.~non-negative random variables with mean and variance equal to one. In parallel with (\ref{eq: batch-sgd}), with $\wh{\beta}^*_0\equiv \wh{\beta}_0$, upon receiving $D_i$, we recursively updates randomly perturbed  SGD estimates,
\begin{eqnarray}\label{psgd1}
s_i^*(\wh{\beta}^*_{i-1})&=&\frac{1}{k}\sum_{l=1}^{k}\sum_{j=1}^{k}\delta_{i,l}\omega_{i}\left(X_{i,l}-X_{i,j}\right) I{\left\{e_{i,l}(\wh{\beta}^*_{i-1}) \leqslant e_{i,j}(\wh{\beta}^*_{i-1})\right\}},\label{SGD-score-wt}\\
\label{psgd2}
\wh{\beta}^*_i&=&\wh{\beta}^*_{i-1}-\gamma_i s_i^*(\wh{\beta}^*_{i-1}),\label{SGD-wt}\\
\label{psgd3}
\ol{\beta}^*_i&=&\frac{1}{i}\sum_{j=1}^i\wh{\beta}^*_j.\label{SGD-avg-wt}
\end{eqnarray}
Recall that the observed data $\mathcal{D}_N=\{Z_1,...,Z_N\}$. 
We will show that $\sqrt{n}(\ol{\beta}_n-{\beta}_0)$ and $\sqrt{n}(\ol{\beta}^*_n-\ol{\beta}_n)|\mathcal{D}_N$ converge in distribution to the same limiting distribution. In practice, these results allow us to estimate the distribution of $\sqrt{n}(\ol{\beta}_n-{\beta}_0)$ by generating a large number, say $B$, of random samples of $\Omega$. We obtain  $\ol{\beta}^{*,b}_n$ by sequentially updating perturbed SGD estimates for each sample $\Omega^b, b=1, \dots, B$,
\begin{eqnarray}
s_i^*(\wh{\beta}^{*,b}_{i-1})&=&\frac{1}{k}\sum_{l=1}^{k}\sum_{j=1}^{k}\delta_{i,l}\omega_{i}^b\left(X_{i,l}-X_{i,j}\right) I{\left\{e_{i,l}(\wh{\beta}^{*,b}_{i-1}) \leqslant e_{i,j}(\wh{\beta}^{*,b}_{i-1})\right\}},\label{SGD-score-wt-b}\\
\wh{\beta}^{*,b}_i&=&\wh{\beta}^{*,b}_{i-1}-\gamma_i s_i^*(\wh{\beta}^{*,b}_{i-1}),\quad i=1,...,n\label{SGD-wt-b}\\
\ol{\beta}^{*,b}_n&=&\frac{1}{n}\sum_{i=1}^n\wh{\beta}^{*,b}_i,\label{SGD-avg-wt-b}
\end{eqnarray}
and then approximate the sampling distribution of $\ol{\beta}_n-\beta_0$ using the empirical distribution of $\{\ol{\beta}^{*,b}_n-\ol{\beta}_n, b = 1, ..., B\}$. Specifically, the covariance matrix of $\ol{\beta}_n$ can be estimated by the sample covariance matrix constructed from $\{\ol{\beta}^{*,b}_n,b = 1,...,B\}$. Estimating the distribution of $\sqrt{n}(\ol{\beta}_n-{\beta}_0)$ based on the distribution of $\sqrt{n}(\ol{\beta}^*_n-\ol{\beta}_n)|\mathcal{D}_N$ leads to the construction of $(1-\alpha)100\%$ confidence regions for $\beta_0$. The resulting inferential procedure retains the numerical simplicity of the SGD method, only using one pass over every $k$ data points. The proposed inferential procedure scales well for datasets with millions of data points or more, and its theoretical validity can be justified  with mild regularity conditions as shown in the next section.
\subsection{Theoretical justification}
In this section, we derive some theoretical properties of $\ol{\beta}_n^*$, justifying that the conditional distribution of $\ol{\beta}_n^*-\ol{\beta}_n$ given data $\mathcal{D}_N=\{Z_1, Z_2, \dots, Z_N\}$ can approximate the sampling distribution of $\ol{\beta}_n-\beta_0$, under the following assumptions. Let $\|\cdot\|$ be the Euclidean norm for vectors and the operator norm for matrices. 
 We first derive the asymptotically linear representation of $\ol{\beta}_n^*$ for any perturbation variables that are i.i.d.~random variables satisfying that $\mathbb{E}(\omega_{i})=1$.
\begin{thm}\label{th:rep}
 If Assumptions A1-A4 hold, and the perturbation variables, $\omega_{1},\cdots, \omega_{
 	n}$, are non-negative i.i.d.~random variables satisfying that $\mathbb{E}(\omega_{i
 })=1$, then we have,
	\begin{equation}\label{eq-th1}
	\sqrt{n}(\ol{\beta}_n^*-\beta_0)=-\frac{1}{\sqrt{n}}H_0^{-1}\sum_{i=1}^n s_i^*(\beta_0)+o_p(1).
	\end{equation}
\end{thm}
\noindent By Theorem \ref{th:rep}, letting $\omega_{i
}\equiv 1$, we derive the following representation for $\ol{\beta}_n$,
\begin{equation}\label{eq-W-equal-one}
\sqrt{n}(\ol{\beta}_n-\beta_0)=-\frac{1}{\sqrt{n}}H_0^{-1}\sum_{i=1}^n s_i(\beta_0)+o_p(1).
\end{equation}
Then, considering the difference between (\ref{eq-th1}) and (\ref{eq-W-equal-one}), we have
\begin{equation}\label{eq-W-diff}
\sqrt{n}(\ol{\beta}_n^*-\beta_n)=\frac{1}{\sqrt{n}}H_0^{-1}\sum_{i=1}^n \{s_i(\beta_0)-s_i^*(\beta_0)\} +o_p(1).
\end{equation}

Let $\mathbb{P}^*$ and $\mathbb{E}^*$ denote the conditional probability and expectation given the data. Starting from (\ref{eq-W-diff}), we derive the following theorem.

\begin{thm}\label{th:dis}
 If Assumptions A1-A4 hold and the perturbation variables, $\omega_{1},\cdots, \omega_{
 	n}$, are non-negative i.i.d.~random variables satisfying that $\mathbb{E}(\omega_{i
 })={\normalfont\mbox{Var}}(\omega_i)=1$, then we have
	\begin{equation}\label{eq-th2}
	\sup_{v\in\mathcal{R}^p}\left|\mathbb{P}^*\left(\sqrt{n}(\ol{\beta}^*_n-\ol{\beta}_n)\leq v\right)-\mathbb{P}\Big(\sqrt{n}(\ol{\beta}_n-\beta_0)\leq v\Big)\right| \rightarrow 0, {in\ probability. }
	\end{equation}
\end{thm}	
\noindent By Theorem \ref{th:dis}, the Kolmogorov-Smirnov distance between $\sqrt{n}(\ol{\beta}^*_n-\ol{\beta}_n)|\mathcal{D}_N$ and $\sqrt{n}(\ol{\beta}_n-\beta_0)$ converges to zero in probability. This validates our proposal of the perturbation-based resampling procedure for inference with the proposed SGD estimator 
$\ol{\beta}_n$.

\section{Numerical studies}

\subsection{Simulation studies}\label{simstudy}
Extensive simulation studies are conducted to assess the operating characteristics of
the proposed SGD methods. We generate the failure time from the model 
\begin{equation}\label{simu}
\log T=X^T\beta+\varepsilon,
\end{equation} 
where $X=\left(X_{1}, \ldots, X_{p}\right)^{T} \in \mathbb{R}^{p}$ follows a multivariate normal distribution with mean zero and the covariance matrix $\operatorname{Cor}(X_j,X_k)=0.3^{|j-k|}$ for $j,k=1, \cdots, p$. We let ${\beta}=\mathbf{1}_{p}$, and $\varepsilon$
follows the standard normal, logistic or extreme-value distribution. The 
censoring time is generated from the uniform distribution to yield 
a censoring proportion of $30\%$. We consider the learning rate $\alpha=0.7$, the sample size $N=50000, 100000$ and let $k=10,50,100$
to examine how the proposed procedures are influenced by different choices of $k$ 
in practice.  For each simulation setting, we repeat the data generation $1000$ times.
 For each data repetition, we use $\Omega=\left\{\omega_{1,
	}, \cdots, \omega_{
	n}\right\}$ as random weights and generate $B=200$ copies of random weights
from the standard exponential distribution. 
 Then, for each data repetition, we take the initial value $\wh{\beta}_{0}=\mathbf{0}_p$,  obtain the proposed SGD estimate (\ref{SGD-avg}) and apply the online resampling procedure to construct $95\%$ confidence intervals.  We report the bias (Bias), standard error (EmSd),
 average estimated standard error (AvSd) and the empirical coverage probability (CovP) of interval estimation at 95\% confidence level. The simulation results for $p=3$ are summarized in Table \ref{tab:p3c30}. We find that
the estimation is quite accurate, the estimation accuracy is robust to varying choices of $k$, the empirical coverage probabilities are close to the nominal level $95\%$, and the empirical standard errors are close to the average estimated standard errors. We also find that the proposed method works  well when the covariates are highly correlated and is insensitive to difference choices of initial estimates or different resampling sizes $B$. This indicates the good performance of the proposed SGD-based estimation and inference procedures.
To examine how the performance of the proposed method varies with the number of covariates $p$, we also let $p=16$ and $p=100$ and the results for 
the first three regression coefficients are summarized in Table \ref{tab:p16a100}. 
 We see that the proposed method works well for both $p=16$ and $p=100$.
For $p=100$, we also report the average time per simulation. We find that
with a small $k$, the savings in computational time are quite substantial and hence the advantages of using the proposed method with a small $k$ are more pronounced when 
$p$ is larger.

\begin{table}
	\caption{Simulation results with $p=3$:	Bias ($\times 10^{-2}$), EmSd ($\times 10^{-2}$), and AvSd ($\times 10^{-2}$) 
		\label{tab:p3c30}}
	{ \footnotesize \begin{center}
		\begin{tabular}{cccrrrrrrrr}
			\hline
	\multirow{2}{*}{$\varepsilon$} 	&	\multirow{2}{*}{Parameter}  &\multirow{2}{*}{$k$}&\multicolumn{4}{c}{$N=50000$}&\multicolumn{4}{c}{$N=100000$}\\ 
			\cline{4-11}
		&	&&Bias &  EmSd&AvSd&CovP&Bias & EmSd&AvSd &CovP\\
			\hline
	Normal&		$\beta_1$			&10&-0.013&0.647&0.653&0.954&-0.022&0.444&0.458&0.947\\
		&	&50&0.020&0.580&0.590&0.953&0.012&0.409&0.416&0.945\\
		&	&100&0.078&0.577&0.589&0.955&0.043&0.413&0.412&0.945\\
			
			&$\beta_2$&	10&-0.017&0.661&0.681&0.955&0.047&0.471&0.481&0.957\\
			&&50&0.032&0.614&0.618&0.950&0.052&0.434&0.434&0.959\\
			&&100&0.129&0.613&0.617&0.946&0.098&0.431&0.431&0.941\\
			
			&$\beta_3$&	10&-0.021&0.632&0.650&0.952&0.013&0.447&0.458&0.953\\
			&&50&0.004&0.591&0.590&0.949&0.028&0.415&0.415&0.949\\
			&&100&0.070&0.588&0.589&0.950&0.062&0.416&0.412&0.933\\
		Logistic&	$\beta_1$
			&10&-0.025&1.039&1.061&0.963&-0.045&0.726&0.739&0.946\\
	&	&50&0.018&0.965&0.972&0.946&-0.033&0.689&0.676&0.942\\
	&		&100&0.023&0.959&0.964&0.950&-0.007&0.681&0.672&0.942\\
			
	&		$\beta_2$&
			10&-0.017&1.061&1.120&0.959&0.013&0.778&0.775&0.953\\
	&		&50&0.016&0.992&1.018&0.958&0.030&0.710&0.710&0.951\\
	&		&100&0.096&0.993&1.011&0.958&0.064&0.706&0.704&0.949\\
			
	&		$\beta_3$&
			10&-0.095&1.054&1.063&0.955&-0.022&0.739&0.738&0.944\\
	&		&50&-0.068&0.978&0.972&0.953&-0.012&0.683&0.677&0.942\\
	&		&100&0.016&0.969&0.962&0.950&0.010&0.680&0.672&0.947\\
	Extreme value&		$\beta_1$
			&10&-0.008&0.812&0.816&0.945&-0.002&0.576&0.567&0.948\\
			&&50&0.030&0.741&0.743&0.960&0.046&0.517&0.518&0.939\\
			&&100&-0.016&0.745&0.739&0.951&-0.022&0.510&0.513&0.946\\
			
			&$\beta_2$&
			10&0.017&0.865&0.855&0.950&0.015&0.583&0.593&0.953\\
		&	&50&0.075&0.782&0.777&0.946&0.054&0.527&0.541&0.949\\
		&	&100&-0.000&0.781&0.773&0.955&-0.011&0.523&0.538&0.942\\
			
		&	$\beta_3$&
			10&0.071&0.794&0.817&0.946&0.049&0.573&0.528&0.948\\
		&	&50&0.167&0.730&0.745&0.940&0.098&0.514&0.518&0.945\\
		&	&100&0.061&0.729&0.737&0.952&0.028&0.507&0.514&0.945\\
			\hline
		\end{tabular}
	\end{center}
}
\end{table}
\begin{table}
	\caption{Simulation results for the first three regression coefficients with $p=16$ or $p=100$: Bias ($\times 10^{-2}$), EmSd ($\times 10^{-2}$),
		and Time (seconds) 
		\label{tab:p16a100}}
	\begin{center}
		\begin{tabular}{cccrrrrrr}
			\hline
			$p$ &		Parameter  &$k$&\multicolumn{3}{c}{$N=50000$}&\multicolumn{3}{c}{$N=100000$}\\ 
			\hline
			16&	&&Bias &  EmSd&CovP&Bias & EmSd &CovP\\
			
			&			$\beta_1$
			&10&-0.088&0.672&0.963&-0.047&0.480&0.954\\
			&			&50&-0.119&0.612&0.937&0.050&0.440&0.952\\
			&			&100&-0.129&0.605&0.944&-0.055&0.430&0.943\\
			
			&			$\beta_2$&
			10&-0.047&0.745&0.960&-0.001&0.484&0.956\\
			&	&50&-0.050&0.651&0.952&-0.005&0.423&0.967\\
			&	&100&-0.040&0.644&0.954&-0.007&0.419&0.960\\
			
			&	$\beta_3$&
			10&0.118&0.697&0.967&-0.013&0.522&0.940\\
			&	&50&-0.143&0.633&0.954&-0.029&0.464&0.937\\
			&	&100&0.136&0.623&0.958&-0.030&0.452&0.935\\
			100&	&&Bias &  EmSd&Time&Bias & EmSd &Time\\
			&			$\beta_1$
			&50&-0.133&0.671&38.361&-0.072&0.445&67.779\\
			&	&200&-0.161&0.653&171.305&-0.084&0.436&341.472\\
			&	&500&0.144&0.675&311.021&-0.071&0.444&618.769\\
			
			&	$\beta_2$
			&50&-0.064&0.713&&-0.033&0.468&\\
			&	&200&0.056&0.687&&-0.027&0.463&\\
			&	&500&0.025&0.709&&0.010&0.462&\\
			
			&$\beta_3$
			&50&-0.069&0.698&&-0.029&0.463&\\
			&	&200&-0.048&0.675&&-0.032&0.448&\\
			&	&500&0.035&0.693&&0.010&0.453&\\
			\hline		
			
		\end{tabular}
	\end{center}
\end{table}
We also compare the proposed rank-based SGD method, denoted by RSGD, with two alternative methods. One is the divide-and-conquer method, denoted by DC Rank,  which splits $N=nk$ observations into $n$ datasets each consisting of $k$ observations. The batch-based Gehan rank
estimation method is applied to each of $n$ datasets, the estimates 
are obtained by using the R package {\it aftgee} \citep{chiou2014fitting}, and $n$
estimates are then averaged to be the final estimator. The other method is to directly apply 
the SGD method to the parametric AFT model, denoted by PSGD, which
approximates the hazard function of $\varepsilon$ in model (\ref{aft}) by a piecewise constant function with $G$ grid points. For PSGD, we choose the step size $\gamma_1=1/50$, the number of
grid points $G=100, 200, 500$ or 
$1000$,  and take the initial estimates of the piecewise hazard function to be
obtained by assuming and fitting a Weibull regression model to the first $m=1000$ data points.
With $N=100,000$, standard extreme value error, censoring
proportion $20\%$, $p=3$ or $p=16$, simulation 
results are summarized in Table \ref{tab:comparision}. 
From the results in Table \ref{tab:comparision}, we can see that
for small or moderate $k$, DC rank exhibits large biases. For DC rank to exhibit 
a small bias, a large $k$ is needed and consequently, it will be  time-consuming to obtain the estimates. The computational inefficiency for DC rank with large $k$ is more pronounced when $p$ is larger. Comparatively, RSGD performs quite well with small $k$ in both of estimation accuracy and computational time. We also find that PSGD is computationally efficient 
but exhibit larger biases than RSGD. Because the iterative estimation in PSGD involves both
the regression coefficients and the piecewise hazard function, PSGD is 
quite sensitive to the number of grid points, the step size, and 
the initial estimates of the piecewise hazard function.  Note that the initial
estimates for the piecewise hazard function obtained by fitting
a Weibull regression model to the first $m=1000$ data points are quite 
good initial estimates since with the standard extreme value error, Weibull regression model
actually holds. In particular, when the initial estimates for 
the piecewise hazard function are not good, we find that PSGD exhibits quite large biases and does not perform well. Its theoretical investigation is also complicated by the choice
of grid points for approximating the baseline hazard function and the fact that
the resulting likelihood is non-smooth and may exhibit multiple local maxima \citep[p.1388]{zeng2007efficient}. In comparison with PSGD, RSGD only involves the regression coefficients in iterative estimation and its performance is quite good and insensitive to the initial values of the regression coefficients, step size $\gamma_1$ or learning rate $\alpha$. It is also important to note that under model misspecification when the AFT model (\ref{aft})
does not hold, RSGD still converges to the minimizer of the convex function $L(\beta)$, further illustrating its robustness.

\begin{table}
	\caption{Comparison of RSGD, DC Rank and PSGD
		with $N=100000$, standard extreme value 
		error and censoring proportion 20\%: Bias ($\times 10^{-2}$), EmSd ($\times 10^{-2}$),
		and Time (seconds) 
		\label{tab:comparision}}
	\begin{center}
{\footnotesize		\begin{tabular}{ccrrrrrrcrrr}
			\hline
			\multirow{2}{*}{Parameter} &\multirow{2}{*}{$k$}&\multicolumn{3}{c}{RSGD}&\multicolumn{3}{c}{DC Rank}&\multirow{2}{*}{$G$}&\multicolumn{3}{c}{PSGD}\\ 
			&&Bias &EmSd  &Time &Bias &  EmSd  &Time	&&Bias &EmSd  &Time\\
			\hline
			\multicolumn{12}{c}{$p=3$}\\
			$\beta_1$			&10&-0.016&0.508&2.110&8.282&21.234&41.062&100&0.081&3.395&6.833\\
			&50&-0.003&0.466&1.569&1.223&0.468&14.755&200&-0.019&3.098&9.364\\
			&100&0.027&0.461&2.002&0.603&0.445&17.022&500&-0.038&3.337&12.930\\
			&10000&-&-&-&-0.068&0.439&1756.109&1000&0.025&3.242&17.838\\
			$\beta_2$
			&10&0.005&0.559&&9.669&36.704&&100&0.432&7.349&\\
			&50&0.030&0.498&&1.283&0.483&&200&-0.017&2.272&\\
			&100&0.068&0.493&&0.635&0.466&&500&-0.058&2.180&\\
			&10000&-&-&-&-0.024&0.471&&1000&0.081&2.208&\\
			$\beta_3$
			&10&-0.013&0.505&&6.751&38.198&&100&-0.056&2.141&\\
			&50&0.003&0.463&&1.259&0.501&&200&-0.069&1.334&\\
			&100&0.034&0.454&&0.630&0.480&&500&-0.020&1.816&\\
			&10000&-&-&-&-0.032&0.449&&1000&-0.092&1.803&\\	
			\multicolumn{12}{c}{$p=16$}\\
			$\beta_1$
			&50&-0.054&0.457&6.447&1.406&0.582&72.730&100&0.118&3.826&7.333\\
			&100&-0.059&0.455&11.737&0.424&0.488&92.142&200&0.028&3.460&10.462\\		
			&200&-0.056&0.453&36.725&0.194&0.450&148.569&500&0.034&3.252&13.758\\		
			&&&&&&&&1000&0.148&3.302&16.325\\
			$\beta_2$
			&50&-0.015&0.481&&1.071&0.622&&100&0.196&8.180&\\
			&100&-0.012&0.474&&0.452&0.498&	&200&0.250&1.607&\\		
			&200&0.005&0.471&&0.211&0.465&&500&-0.024&2.086&\\	
			&&&&&&&&1000&0.081&2.187&\\	
			$\beta_3$
			&50&-0.026&0.493&&1.118&0.630&&100&0.048&2.457&\\
			&100&-0.020&0.486&&0.482&0.518&&200&0.102&1.846&\\		
			&200&-0.003&0.481&&0.241&0.481&&500&0.046&1.732&\\	
			&&&&&&&&1000&0.050&1.673&\\	
			\hline

		\end{tabular} }
	\end{center}
\end{table}

	%
	%

			
			
Next, we examine the computational scalability of the proposed method and compare it with the batch-based method.  Note that the computational complexity of the batch-based method is $O(N^2)$.
By \cite{jin2003rank}, the optimization of (\ref{obj:gehan}) can be formulated as a linear programming problem and 
we 
use the the R package {\it aftgee} \citep{chiou2014fitting} to  obtain the estimator. We let the error $\varepsilon$ follows the 
standard extreme value distribution and the censoring proportion is set to be $20\%$. The data are generated as before and we repeat the data generation $200$ times. We use $\mbox{Batch}$ to denote the classical  batch-based method and use $\mbox{RSGD}(k)$ to denote the proposed SGD method which updates the estimates every $k$ data points. The average computation time   
per simulation is summarized in Table \ref{tab:time} for varying $N$ and $k$. From Table \ref{tab:time}, we can see that for the classical batch-based method, the computational 
time increases linearly with $N^2$ and hence becomes memory inefficient or computationally prohibitive when $N$ is very large. However, the proposed SGD method scales well with the size of the dataset and is computationally and memory efficient, in particular, 
with choices of small $k$.

\begin{table}[H]
	\caption{The average computation time (seconds) per simulation with varying $N$ and $k$:
	$p=3$ and censoring proportion 20\%}\label{tab:time}	
	\begin{center}
		\begin{tabular}{cccccccc}
			\hline
			\multirow{2}{*}{Method} &&
			\multicolumn{6}{c}{$N$}\\
			\cline{3-8}
			&&1000&2000&5000&10000&20000&50000\\
			\hline
			RSGD(50)&&0.0325& 0.0252& 0.1195& 0.0566&0.2555& 0.8479\\
			RSGD(100)&&0.0733& 0.0714& 0.0914&0.1465&0.6175& 2.3856\\
			RSGD(200)&&0.1879&0.1370& 0.1865&0.3778& 1.0671&  4.7697\\
			RSGD(500)&&0.3275&0.3023& 0.4059& 0.6480&1.9212&  9.3798\\
			RSGD(1000)&&0.7670& 0.6207& 1.0716&1.7229&5.2440&22.7282\\
			Batch&&1.1157&4.4077&26.7893& 104.4491&696.4405&2970.832\\			
			\hline
		\end{tabular}			
	\end{center}
\end{table}

Finally, we evaluate the relative  efficiency of the proposed SGD estimation method to the batch-based estimation method (\ref{obj:gehan}). 
Because when $N$ is large, it is time-consuming to obtain the batch-based estimates,  we use the asymptotic relative efficiency formula $RE(k)$ in Section \ref{are} for the assessment. Because the formula 
involves the unknown density function and
the censoring distribution, we use the Monte-Carlo method to evaluate it and assess the impact of $k$ on the asymptotic relative efficiency. We consider the 
situation where $p=2$, $X_1$ and $X_2$ are independent 
standard normal, and the censoring proportion is $20\%$.
The results  are summarized in Table \ref{eff}. We can see that the performance 
of the proposed SGD method is quite robust to varying choices of $k$ if gauged by 
the asymptotic relative efficiency and when $k$ is moderately large such as $100$ or 
$200$, $RE(k)$ is close to $1$. This affirms our theoretical result in Section \ref{are} and indicates that in addition to its superior computational advantages,
the proposed method performs also well in terms of the estimation efficiency.


\begin{table}
	\caption{ Monte-carlo evaluation of the asymptotic relative efficiency $RE(k)$ 
		\label{eff}}
	\begin{center}
		\begin{tabular}{ccccrrrrr}
			\hline
			\multirow{2}{*}{$\varepsilon$}&\multirow{2}{*}{Parameter}  &\multirow{2}{*}{N}&&
			\multicolumn{5}{c}{$k$}\\
			\cline{5-9}
			&&&&10&20&50&100&200\\
			\hline
			Normal&$\beta_1$&50000&& 0.84544& 0.88343&0.93028& 0.90950& 0.91178\\
			&&100000&&0.88871&0.91784&0.94116&0.96727&0.96564\\ 
			
			&$\beta_2$&50000&& 0.84718& 0.90223& 0.91407& 0.93041&0.93866\\
			&&100000&&0.89221 &0.95048& 0.98471& 0.99638& 0.99110 \\
			\\
			Logistic&$\beta_1$&50000&&  0.74719& 0.77522& 0.80045& 0.80507& 0.80818\\
			&&100000&& 0.86510 &0.89969& 0.91819& 0.93230& 0.92696\\

			&$\beta_2$&50000&& 0.71640& 0.76113& 0.78608& 0.80496& 0.81454\\
			&&100000&& 0.81932& 0.86552& 0.90754& 0.90614& 0.92480\\
			
			\\
			Extreme-value	&$\beta_1$&50000&&0.85502& 0.90839& 0.94051& 0.95154& 0.95391\\
			&&100000&& 0.87155& 0.92505& 0.96066& 0.96950& 0.97452\\
			
			&$\beta_2$&50000&& 0.85966& 0.91430& 0.94762& 0.97138& 0.97044\\
			&&100000&& 0.86105& 0.89223& 0.92108& 0.92361& 0.92038\\		
			\hline
		\end{tabular}			
	\end{center}
\end{table}

\subsection{An application to the SEER breast cancer data}
We applied the proposed method to the breast cancer data collected in the U.S. National Cancer Institute's Surveillance, Epidemiology, and End Results-SEER Public‐Use Database. The SEER program collects cancer incidence and survival data from population-based cancer registries covering approximately 34.6\% of the population of the United States. The database records a number of patient and tumor characteristics such as demographics, primary tumor site, tumor morphology, stage at diagnosis, and first course of treatment, and the registries follow up with patients for vital status. With such a large-scale dataset, it offers a unique opportunity to examine the effect of patient and tumor characteristics on survival which is of most relevance in many cancer survival studies \citep{rosenberg2005effect,mccready2000factors}.

We consider the data collected from 1998 to 2015 and select the data by the following criteria: 1) Races of either white or black; 2) Tumor size of less than 2cm; 
3) Cancer stages of In situ, Localized, Regional, Distant;
and 
4) No missing values. The dataset consists of  $221,013$ observations. The 
response of interest is time to death due to breast cancer and the censoring proportion is about 90\%. The model (\ref{aft}) is employed to examine the covariate effects and 
the proposed SGD methods are used for estimation (\ref{SGD-avg}) and inference (\ref{SGD-avg-wt-b}). 
The following covariates are included in our analysis.  1) Age at diagnosis ( 6 levels: Below 35, 36-45, 46-55, 56-65, 66-75, Above 75); 2) Race ( 1 for White and 0 for Black); 3) Tumor Grade ( 4 levels: Well differentiated, Moderately differentiated. Poorly differentiated, Undifferentiated); 4) Cancer Stage (4 levels: In situ, Localized, Regional, Distant); 5) Year of diagnosis (3 levels: 1997-2003, 2004-2009, 2010-2015); and 6) the logarithm of Tumor size (mm). The categories ``Above 75", ``Undifferentiated", ``Distant" are taken to be the reference levels for Age, Tumor Grade and Cancer Stage, respectively. 


The estimated regression coefficients along with their 95\% confidence intervals with varying choices of $k$ are reported in Table \ref{seer}.  
We find that the survival is longer for women of medium age 36-45,
compared with that of younger or older women. This result is consistent with 
previous works by \cite{wingo1998long} and \cite{rosenberg2005effect}.
The effect of Race on survival is also significant and consistent with previous studies \citep{li2003differences}. 
For Tumor Grade, we see that
patients with smaller grade levels tend to live longer. For Cancer Stage, patients 
who are diagnosed at an earlier stage have a larger chance of being cured, especially
for the In situ stage. We also find that Year at diagnosis has 
a significant effect on survival, which indicates that the effectiveness of treatment 
improves over time along with advances of medical research. For 
Tumor size, the effect is quite pronounced and the tumor size is negatively 
correlated with the survival. It is important to note that for varying choices of $k$,
the proposed method yields fairly robust point and interval estimates for the regression coefficients.
This reaffirms our findings in simulation studies and demonstrates that the proposed estimation and inference procedures indeed provide a useful tool for  large-scale or online analysis of survival data which is becoming increasingly prevalent in practice.
\begin{table}
	\small\addtolength{\tabcolsep}{-5pt}
	
	\caption{The estimated regression coefficients and 95\% confidence interval (in parentheses)  for the  SEER breast cancer data  \label{seer}}
	\begin{center}
		\begin{tabular}{lcccc}
			\hline 
			Covariate&$k=10$&$k=20$&$k=50$&$k=100$\\
			\hline
			\multirow{2}{*}{Age (Below 35)} &0.4352& 0.4357&  0.4336&  0.4325\\
		&  (0.4077,  0.4628)&  (0.4054, 0.4659)& (0.4046, 0.4626)&  (0.4056, 0.4593)\\
		
		\multirow{2}{*}{Age (36-45)} &0.7345 & 0.7290&  0.7224 &  0.7251\\
		& (0.6868,  0.7822)  & (0.6902,  0.7678) & (0.6834, 0.7615)& (0.6906,  0.7595)\\
		
		\multirow{2}{*}{Age (46-55)} &0.6931& 0.7052&  0.7099& 0.7082\\
		&(0.4990,  0.8872)&  (0.6061,  0.8042)& (0.6497,  0.7701)&  (0.6607,  0.7556)\\
		
		\multirow{2}{*}{Age (56-65) }& 0.7117&  0.7173 & 0.7179&  0.7090\\
		&(0.6411,  0.7823)  & (0.6640,  0.7706) &(0.6683,  0.7674) &  (0.6673, 0.7506)\\
		
		\multirow{2}{*}{Age (66-75)}&  0.4026& 0.4037& 0.4068&  0.4030\\
		&  (0.3397,  0.4654 )&  (0.3416,  0.4659)& (0.3475,  0.4662) &  (0.3496,  0.4564)\\
		
		\multirow{2}{*}{Race}  &   0.4665& 0.4712& 0.4593&  0.4539\\
		&(0.2491,  0.6839) & ( 0.3332,0.6092)&   (0.3604,  0.5583)&  (0.3846, 0.5232)\\
		
		\multirow{2}{*}{Grade (Well differentiated)}&0.8245&  0.8270&  0.8445&  0.8391\\
		&(0.7568,  0.8921)&  (0.7777, 0.8764)&(0.7954,  0.8936) &  (0.7949,  0.8833)\\
		
		\multirow{2}{*}{Grade (Moderately differentiated)}& 0.5662&  0.5722& 0.5679&  0.5712\\
		&(0.5189,  0.6134) & (0.5312,  0.6132)&(0.5265,  0.6093)&  (0.5334, 0.6090)\\
		
		\multirow{2}{*}{Grade (Poorly differentiated)}&-0.0173& -0.0181& -0.0242& -0.0191\\
		&(-0.0688, 0.0341)& (-0.0545, 0.0184)&  (-0.0612, 0.0129)& (-0.0547,  0.0165)\\
		
		\multirow{2}{*}{Stage (In situ)} & 2.9620&  2.9605& 2.9653& 2.9629\\
		&(2.9047, 3.0192)& (2.9176, 3.0034)&(2.9299, 3.0007) & (2.9356,  2.9901)\\
		
		\multirow{2}{*}{Stage (Iocalized)} &  2.4285& 2.4393&  2.4235&2.4184\\
		&(2.3187, 2.5383) &  (2.3611, 2.5176)&(2.3613, 2.4857)&  (2.3616, 2.4751)\\
		
		\multirow{2}{*}{Stage (Regional)} & 1.7096& 1.714&  1.7174&  1.7220\\
		& (1.5626, 1.8564)  & (1.6091,1.8182)&(1.6503  , 1.7845) &  (1.6573  , 1.7866)\\
		
		\multirow{2}{*}{Year (1997-2003)} &  -0.1100& -0.0894& -0.1011&-0.1070\\
		&(-0.2329, 0.0129) &(-0.1800,  0.0012)&  (-0.1767 , -0.0255) & (-0.1770 ,-0.0369)\\
		
		\multirow{2}{*}{Year (2004-2009) }&-0.0590& -0.0646& -0.0581&-0.0556\\
		&(-0.1815,  0.0636) &(-0.1576,  0.0284)&(-0.1343,  0.0180 )& ( -0.1277 , 0.0165)\\
		
		\multirow{2}{*}{log(Tumor size)} & -0.6777&-0.6797& -0.6863& -0.6946\\
		&(-0.7492, -0.6061)& (-0.7189, -0.6404)&(-0.7252, -0.6473)& (-0.7289, -0.6603)\\		
			\hline
		\end{tabular}
	\end{center}
\end{table} 

\section{Discussion}
We have presented a rank-based stochastic gradient descent (RSGD) method for large-scale or online survival data. The estimate is recursively updated upon the arrival of every $k$ observations.  Alternatively, for the divide and conquer (DC) approach,  $N=nk$ observations are split into $n$ datasets each consisting of $k$ observations, the batch estimation method is applied to each of $n$ datasets and $n$ estimates are then averaged to be the final estimate. We find that RSGD works well with small $k$ but for the DC estimator to exhibit small biases, large $k$ is needed. This indicates the advantages of RSGD under
memory and time constraints.

The proposed SGD method is built upon rank-based estimating equations which do not directly allow time-dependent covariates. For their incorporation, likelihood-based SGD methods can be considered. We find that the SGD method based on the parametric AFT model that takes the baseline hazard function to be piece-wise constant is quite sensitive 
to the initial estimates. This is consistent with the findings 
of \cite{zeng2007efficient} that its corresponding sieve profile log-likelihood 
is not smooth and may have multiple local maxima.  
A kernel-smoothed profile likelihood based 
SGD method inspired by their work can be explored. In particular,
it is of interest  to investigate how the smoothing 
parameter is chosen for SGD-based methods in both theory and practice.

\bibliographystyle{apa}
\bibliography{qaft}
\end{document}